\documentclass[10pt,conference]{IEEEtran}

\usepackage[letterpaper, bottom=1.08in, top=0.74in, left=0.625in, right=0.625in]{geometry} 

\usepackage{amsmath,graphicx}
\usepackage{bbm}
\usepackage{amsmath, amssymb, amscd, amsthm, amsfonts}
\usepackage{subfiles}
\usepackage{color}

\pdfoutput=1
\usepackage[caption=false,font=footnotesize]{subfig}

\usepackage[backend=biber, style=ieee, dashed=false,isbn=false,doi=false,url = false,bibencoding=utf8,citestyle=numeric-comp]{biblatex}
\usepackage[linesnumbered,ruled,vlined]{algorithm2e}
\usepackage[noend]{algpseudocode}

\SetKwInput{KwInput}{Input}                
\SetKwInput{KwOutput}{Output}              

\SetCommentSty{mycommfont}

\newtheorem{theorem}{Theorem}
\newtheorem{lemma}{Lemma}

\newtheorem{definition}{Definition}

\title{Fundamental Limits of 1-bit ISAC Systems: Capacity Region and Optimal Power Control}

\author{%
  \IEEEauthorblockN{Emmanuel Trinidad\IEEEauthorrefmark{1}\IEEEauthorrefmark{2}, Neil Irwin Bernardo\IEEEauthorrefmark{1}}
  \IEEEauthorblockA{\IEEEauthorrefmark{1}Electrical and Electronics Engineering Institute, University of the Philippines Diliman, Quezon City, Philippines \\
  \IEEEauthorrefmark{2}Department of Electronics Engineering, Pampanga State University, Pampanga, Philippines }
  \IEEEauthorblockA{
      Email: \{emmanuel.trinidad,  neil.bernardo\}@eee.upd.edu.ph
   }
}

\begin{document}
\maketitle
\begin{abstract}
This paper investigates the fundamental limits of integrated sensing and communication (ISAC) systems with 1-bit receiver quantization. We analyze a Gaussian fading ISAC channel with separate communication and monostatic sensing links, where both communication and sensing receivers are equipped with 1-bit quantizers. When the communication channel state information (CSI) is available at the receiver, we characterize the communication-sensing capacity region of 1-bit ISAC channel and show that no trade-off exists between communication and sensing performance. In particular, both communication and sensing capacities can be simultaneously achieved by a constant-amplitude input distribution with a specific rotational symmetry. For the scenario where communication CSI is also available at the transmitter, we formulate a weighted optimization problem that balances communication and sensing rates in 1-bit ISAC channel under an average power constraint and then derive the corresponding optimal power control policy. The results demonstrate how the optimal power control policy evolves with the weighting parameter, transitioning from a communication-centric, opportunistic transmission to a more uniform allocation as sensing becomes increasingly prioritized.

\end{abstract}
\section{Introduction}
\label{section:intro}

Integrated sensing and communication (ISAC) has emerged as a key enabling technology for 6G wireless systems, aiming to unify sensing and communication functionalities within a common hardware platform, signaling framework, and spectral resources. Traditionally, radar and communication systems have been designed and operated independently. In contrast, 6G envisions tightly integrated dual-function systems capable of simultaneously supporting data transmission and environmental sensing tasks \cite{Gosh:2025}.

From an information-theoretic perspective, recent studies have characterized the fundamental trade-offs between communication and sensing performance in terms of achievable rate and distortion metrics, showing that joint optimization of the input distribution can yield performance gains over radar-centric or communication-centric input design \cite{ahmadipour_2024}. Mutual information (MI) has also been proposed as a unifying performance metric for communication and sensing tasks \cite{Ouyang_2023}. In particular, sensing mutual information (SMI) quantifies the amount of information that the received signal conveys about the underlying environment or target parameters, and shares the same units and key properties as communication mutual information (CMI). This unified MI-based framework enables a direct comparison and joint optimization of communication and sensing objectives within a common information-theoretic formulation. Within this framework, the trade-off between random and deterministic signaling in ISAC systems has also been revealed in \cite{Liu:2023}. Furthermore, waveform design strategies based on weighted combinations of CMI and SMI have been developed for multi-antenna ISAC systems, highlighting the trade-offs between communication and sensing objectives \cite{Piao_2023}.

In parallel with ISAC developments, the increasing demand for high data rates and wider bandwidths has led to significant challenges in hardware cost and power consumption. Low-resolution analog-to-digital converters (ADCs), particularly 1-bit ADCs, have emerged as a promising solution due to their energy efficiency, albeit at the cost of reduced information fidelity \cite{Liu_Jun:2019}. A substantial body of work has characterized capacity and optimal input distributions under quantized outputs across various channel models, including real-valued channels \cite{singh_2009}, complex fading channels \cite{Krone_2010}, non-coherent Rician channels \cite{vuminh2019_optsig}, phase-quantized systems \cite{bernardo2021_ISIT,bernardo2022_phasequant}, and polar-quantized systems \cite{bernardo2022_polarquant}. More recently, low-resolution quantization has also been incorporated into ISAC system design, with optimization frameworks accounting for detection performance and quality-of-service requirements \cite{salman2024_1-bitDAC_ISAC,wang2024massive, Li:2026}. However, despite these advances, the structure of the capacity-achieving input distribution for 1-bit ISAC systems remains largely unexplored, even in the single-antenna setting.

In this paper, we adopt the unified MI-based framework to characterize both communication and sensing performance of 1-bit ISAC systems under varying levels of channel state information (CSI) availability. When the CSI is only available at the communication receiver, we show that both communication and sensing capacities can be simultaneously achieved by a constant-amplitude input distribution with $\frac{\pi}{2}$-rotational symmetry, indicating the absence of the communication and sensing performance trade-off in the 1-bit ISAC setting. When CSI is also available at the transmitter, optimizing a weighted combination of CMI and SMI reveals intriguing dynamics in the optimal power allocation strategy.

The remainder of this paper is organized as follows. Section \ref{section:sys_model} presents the 1-bit Gaussian fading ISAC channel. Section \ref{section:capacity_region} introduces the MI-based framework and characterizes the communication–sensing capacity region when the communication CSI is only available at the communication receiver. Section \ref{subsection:P_opt_formulation}  considers the CSIT scenario and formulates the joint communication–sensing optimization problem, along with the derivation of the optimal power control policy. Section \ref{subsection:numerical_results} provides numerical results that illustrate the structure of the optimal policy and the resulting trade-offs between communication and sensing performance. Finally, Section \ref{section:conclusion} concludes the paper.

\emph{Notation:} We use $\mathbb{C}$, $\mathbb{R}$, and $\mathbb{Z}$ to denote the sets of complex numbers, real numbers, and integers, respectively. Uppercase $X$ denotes a random variable while lowercase $x$ denotes its realization. For any $v\in\mathbb{C}$, $\Re\{v\}$ and $\Im\{v\}$ give the real and imaginary components of $v$. When it is clear from the context, we use $F_{X}$ and $f_{X}$ to denote the cumulative distribution function (CDF) $F_{X}(x)$ and probability density function (PDF) $f_{X}(x)$, respectively. Equivalent notation can be used for the
conditional and joint distributions. We write the probability
mass function (PMF) of a random variable Y as $p_Y(y; F_X)$ if
the distribution is induced by a distribution
$F_X$. For instance, $p_Y(y; F_X) = \int p_{Y|X}(y|x)\;dF_{X}$ for some condition PMF $p_{Y|X}$. We use $\mathbb{H}(\cdot)$ and $\mathbb{I}(*;\cdot)$ to denote the entropy and mutual information, respectively.  The output entropy and conditional output entropy of $Y$ induced by some input distribution $F_X$ are denoted by $\mathbb{H}_{F_X}(Y)$ and $\mathbb{H}_{F_X}(Y|X)$, respectively. All $\log()$ terms are base 2 unless otherwise stated. We define the operator $[x]^+ = \max(0,x)$. We use the operator $\mathrm{sgn}\{x\}$ to extract the sign of $x$. Lastly, we use $\mathcal{N}_{\mathbb{C}}(\mu,\sigma^2)$ to denote a circular-symmetric complex Gaussian (CSCG) distribution with mean $\mu$ and variance $\sigma^2$.


\section{System Model}
\label{section:sys_model}

\begin{figure*}[t!]
\centering
    \includegraphics[width=.95\linewidth]{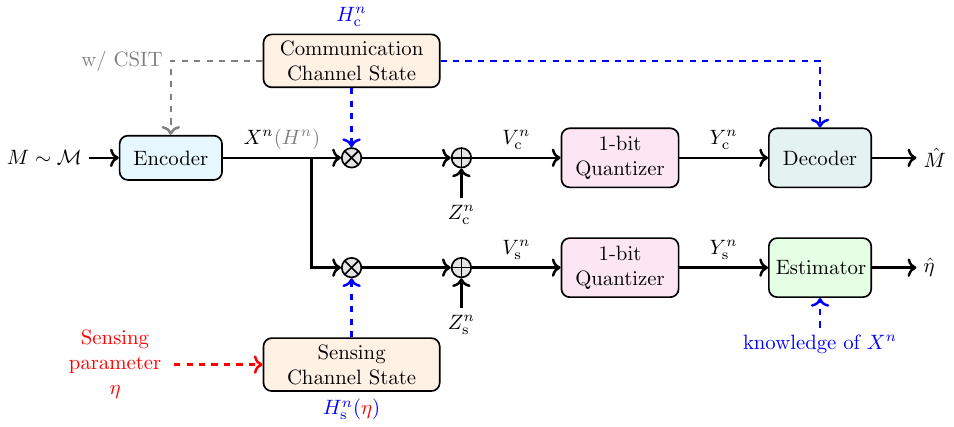}
    \caption{System model of the 1-bit ISAC channel. Two CSI scenarios are considered: (1) CSI is known only at the receiver, and (2) CSI is available at both transmitter and receiver. Under scenario 2, the encoder can adapt its transmission according to the fading process $H_{\mathrm{c}}^n$.}
    \label{fig:isac_monostatic}
\end{figure*}

We consider the 1-bit ISAC monostatic system illustrated in Fig.~\ref{fig:isac_monostatic}. The ISAC transmitter aims to convey a message $M$, drawn uniformly from the message set $\mathcal{M} = \{1,2,\cdots,|\mathcal{M}|\}$, to a communication receiver over a fading channel using $n$ channel uses. The encoder produces a complex-valued codeword, denoted by $X^n \in \mathbb{C}^n$, that satisfies the average power constraint
\begin{align}
    \frac{1}{n}\sum_{i = 1}^{n}|X_i|^2 \leq P.
\end{align}
This codeword undergoes an i.i.d.\ fading process $H_{\mathrm{c}}^n$ and is corrupted by an additive noise $Z_{\mathrm{c}}^n$. Specifically, the fading coefficients $\{H_{\mathrm{c},i}\}_{i = 1}^{n}$ and the noise samples $\{Z_{\mathrm{c},i}\}_{i = 1}^{n}$ are assumed to be CSCG, i.e., $H_{\mathrm{c},i}\sim\mathcal{N}_{\mathbb{C}}(0,1)$ and $Z_{\mathrm{c},i}\sim\mathcal{N}_{\mathbb{C}}(0,\sigma_{\mathrm{c}}^2)$.


At the same time, the reflected signal corresponding to the transmitted codeword $X^n$ is observed by a monostatic sensing receiver, which aims to estimate a parameter of interest (POI) $\eta$ that characterizes the wireless environment. We assume that there is a one-to-one mapping between the POI and the i.i.d. fading process, denoted by $H_{\mathrm{s}}^n$, such that all uncertainty on $\eta$ is captured by $H_{\mathrm{s}}^n$ and has perfect self interference cancellation (SIC). Similar to the communication channel, the fading coefficients and the noise samples of the sensing channel are assumed to be CSCG, i.e., $H_{\mathrm{s},i}\sim \mathcal{N}_{\mathbb{C}}(0,1)$ and $Z_{\mathrm{s},i}\sim \mathcal{N}_{\mathbb{C}}(0,\sigma_{\mathrm{s}}^2)$.

The resulting unquantized input-output relationships of the ISAC channel can be written as
\begin{align}\label{eq:unquantized_sig}
    V_{\mathrm{c},i} =& H_{\mathrm{c},i}\cdot X_i + Z_{\mathrm{c},i}\\
    V_{\mathrm{s},i} =& H_{\mathrm{s},i}\cdot X_i + Z_{\mathrm{s},i},
\end{align}

In this work, we focus on the regime where both the communication receiver and the monostatic sensing receiver are equipped with 1-bit quantizers operating independently on the in-phase and quadrature components. Accordingly, decoding and estimation are performed based on the quantized sequences $Y_{\mathrm{c}}^n$ and $Y_{\mathrm{s}}^n$, respectively, whose $i$-th elements are given by
\begin{align}\label{eq:quantized_sig}
    Y_{\mathrm{c},i} =& \mathcal{Q}_{\mathrm{1bit}}(V_{\mathrm{c},i}) = \mathcal{Q}_{\mathrm{1bit}}(H_{\mathrm{c},i}\cdot X_i+Z_{\mathrm{c},i})\\
    Y_{\mathrm{s},i} =& \mathcal{Q}_{\mathrm{1bit}}(V_{\mathrm{s},i}) = \mathcal{Q}_{\mathrm{1bit}}(H_{\mathrm{s},i}\cdot X_i+Z_{\mathrm{s},i}),
\end{align}
where the 1-bit quantization operator $\mathcal{Q}_{\mathrm{1bit}}\left(\cdot\right)$ is defined as
\begin{align}\label{eq:1bit_quant}
    \mathcal{Q}_{\mathrm{1bit}}\left(\cdot\right) = \mathrm{sgn}\left(\Re\left\{\cdot\right\}\right) + j\cdot\mathrm{sgn}\left(\Im\left\{\cdot\right\}\right).
\end{align}
The output of $\mathcal{Q}_{\mathrm{1bit}}\left(\cdot\right)$ can be one of the following values in the set $\mathcal{Y}$:
\begin{align}
    \mathcal{Y} = \{\underbrace{1+j}_{y_{0}},\;\underbrace{-1+j}_{y_{1}},\;\underbrace{-1-j}_{y_{2}},\;\underbrace{1-j}_{y_{3}}\}.
\end{align}
Intuitively, the 1-bit I/Q quantizer outputs $y_{l} = \mathcal{Q}_{\mathrm{1bit}}(v)$ whenever $v$ falls in the $(l+1)$-th quadrant of the I/Q plane. From this problem description we are led to the following question: \emph{Which pairs of communication and sensing rates are achievable in the 1-bit ISAC channel?}

We answer this question for two CSI scenarios. In the first scenario, we assume CSI at the receiver only (CSIR), where the communication receiver has perfect knowledge of the fading coefficients $H_{\mathrm{c}}^n$ while the transmitter has access only to their statistical distribution. Consequently, the encoder cannot adapt its transmission strategy to the instantaneous channel realizations, and reliable communication must be achieved based solely on receiver-side channel knowledge. Under this scenario, we derive the communication-sensing capacity region of the 1-bit ISAC system and identify the modulation scheme that achieves the Pareto boundary of this region.

In the second scenario, we consider CSI at the transmitter (CSIT), where the transmitter is additionally assumed to know the realizations of the communication channel fading process $H_{\mathrm{c}}^n$. This enables the encoder to adapt the transmitted codeword to the instantaneous channel state. Under this scenario, we investigate how the relative weights between sensing and communication objectives influence the optimal power allocation strategy of the ISAC transmitter.

\section{Communication-Sensing Capacity Region of 1-bit ISAC System under CSIR}
\label{section:capacity_region}

In this section, we first present the probability distributions that are relevant in the formulation of the CMI and SMI of the 1-bit ISAC channel. We then characterize the communication-sensing capacity region of the 1-bit ISAC channel assuming the communication CSI is known at the communication receiver.

\subsection{Communication and Sensing Transition Probabilities}
We first examine the transition probabilities of the communication channel and sensing channel. To simplify the exposition, we introduce the channel identifier $\mathrm{a} \in \{\mathrm{c},\mathrm{s}\}$ to select between communication and sensing channel and consider the input-output relationship
\begin{align}
    Y_{\mathrm{a}} = \mathcal{Q}_{\mathrm{1bit}}\left(V_{\mathrm{a}}\right) = \mathcal{Q}_{\mathrm{1bit}}\left(H_{\mathrm{a}}\cdot X + Z_{\mathrm{a}}\right).
\end{align}

Given $X$ and $H_{\mathrm{a}}$, the (unquantized) channel output $V_{\mathrm{a}}$
is a CSCG whose conditional PDF can be written as
 \begin{align}
    f_{V_{\mathrm{a}}|H_{\mathrm{a}},X}\left(v_{\mathrm{a}}|h_{\mathrm{a}},x\right) = \frac{1}{\pi\sigma_{\mathrm{a}}^2}\cdot\exp\left(-\frac{|v_{\mathrm{a}} - h_{\mathrm{a}}x|^2}{\sigma_{\mathrm{a}}^2}\right),
\end{align}
with $\sigma_{\mathrm{a}}^2$ being either $\sigma_{\mathrm{c}}^2$ or $\sigma_{\mathrm{s}}^2$.
To describe the channel laws under 1-bit I/Q quantization, we first define the function
  \begin{align}\label{eq:gen_channel_law_1bit}
     W^{(\mathrm{a})}_{y_l}(z) = Q\left(-\frac{ \Re\{z\}\Re\{y_{l}\}}{\sigma_{a}/\sqrt{2}}\right) \cdot Q\left(-\frac{ \Im\{z\}\Im\{y_{l}\}}{\sigma_{a}/\sqrt{2}}\right),  
\end{align}
where $z\in \mathbb{C}$, and $Q(\cdot)$ is the Gaussian $Q$-function, i.e., the tail probability of the standard Gaussian distribution. Intuitively, $W^{(\mathrm{a})}_{y_l}(z)$ gives the probability that $Y_{\mathrm{a}} = y_l$ given that $z$ is sent over an additive white Gaussian noise (AWGN) channel with variance $\sigma_{\mathrm{a}}^2$. Note that the two $Q$-function terms correspond to $\mathbb{P}\left(\Re\{Y_{\mathrm{a}}\} = \Re\{y_l\}| Z = z\right)$ and $\mathbb{P}\left(\Im\{Y_{\mathrm{a}}\} = \Im\{y_l\}|Z = z\right)$, respectively. This factorization implies that the $\Re\{Y_{\mathrm{a}}\}$ and $\Im\{Y_{\mathrm{a}}\}$ are conditionally independent given $Z$.

Using \eqref{eq:gen_channel_law_1bit}, the communication channel law and sensing channel law under 1-bit I/Q quantization can be written as $W_{y_l}^{(\mathrm{c})}(h_{\mathrm{c}}\cdot x)$ and $W_{y_l}^{(\mathrm{s})}(h_{\mathrm{s}}\cdot x)$, respectively. Consequently, given some input distribution $F_{X}$, the PMF of the 1-bit communication channel output $Y_{\mathrm{c}}$ (conditioned on $H_{\mathrm{c}}$) is
\begin{align}
    p_{Y_{\mathrm{c}}|H_{\mathrm{c}}}(y_{l}|h_{\mathrm{c}};F_{X}) =& \int_{\mathbb{C}}W_{y_{l}}^{(\mathrm{c})}(h_{\mathrm{c}}\cdot x)\;\mathrm{d}F_{X}
\end{align}
while the PMF of the 1-bit sensing channel output $Y_{\mathrm{s}}$ (conditioned on $X$) is
\begin{align}
    p_{Y_{\mathrm{s}}|X}(y_{l}|x) =& \int_{\mathbb{C}}W_{y_{l}}^{(\mathrm{s})}(h_{\mathrm{s}}\cdot x)\;\mathrm{d}F_{H_{\mathrm{s}}}.
\end{align}

\subsection{Communication and Sensing Mutual Information}

The achievable communication rate $R_{\mathrm{c}}$ of the communication link, assuming CSI is available at the receiver, is described by the CMI whose expression can be written as
\begin{align}
    \mathbb{I}\left(X;Y_{\mathrm{c}}|H_{\mathrm{c}}\right) = \mathbb{H}_{F_{X}}\left(Y_{\mathrm{c}}|H_{\mathrm{c}}\right) - \mathbb{H}_{F_{X}}\left(Y_{\mathrm{c}}|H_{\mathrm{c}},X\right).
\end{align}
For notational convenience, let us define the function $\xi(x) = -x\log x$. Then, the first entropy term can be expressed as
\begin{align}
    \mathbb{H}_{F_{X}}\left(Y_{\mathrm{c}}|H_{\mathrm{c}}\right) = \int_{\mathbb{C}}\sum_{l = 0}^{3}\xi\left(p_{Y_{\mathrm{c}}|H_{\mathrm{c}}}(y_{l}|h_{\mathrm{c}};F_{X})\right)\;\mathrm{d}F_{H_{\mathrm{c}}} \label{eqn:H(Yc|H)}
\end{align}
while the second entropy term can be expressed as
\begin{align}
\mathbb{H}_{F_{X}}\left(Y_{\mathrm{c}}|H_{\mathrm{c}},X\right) = \iint\limits_{\mathbb{C}\;\mathbb{C}}\sum_{l = 0}^{3}\xi\left(W_{y_{l}}^{(\mathrm{c})}(h_{\mathrm{c}} x)\right)\mathrm{d}F_{H_{\mathrm{c}}}\mathrm{d}F_{X}. \label{eqn:H(Y_c|H,X)}
\end{align}
Similarly, we adopt the sensing rate $R_{\mathrm{s}}$ as the performance metric for the sensing task. This sensing rate is described by the SMI, which is the mutual information between the 1-bit sensing channel output $Y_{\mathrm{s}}$ and the fading coefficient of the sensing channel $H_{\mathrm{s}}$. SMI assumes that $\eta$ is random, in contrast to estimation metrics that can cater to both deterministic and random POIs   \cite{Ouyang_2023}. Nevertheless, SMI is agnostic to the estimator and provides a universal lower bound to the distortion metrics. In Gaussian channels, the close connection between estimation metrics and SMI has been proved, making it a suitable measure of sensing performance \cite{Liu:2023,xie:2024_SMI}. The SMI can be written explicitly as
\begin{align}
     \mathbb{I}\left(\eta;Y_{\mathrm{s}}|X\right) =& \mathbb{I}\left(H_{\mathrm{s}};Y_{\mathrm{s}}|X\right)\nonumber\\
     =&\mathbb{H}_{F_{X}}\left(Y_{\mathrm{s}}|X\right) - \mathbb{H}_{F_{X}}\left(Y_{\mathrm{s}}|H_{\mathrm{s}},X\right),
\end{align}
where the first entropy term can be expressed as
\begin{align}
\mathbb{H}_{F_{X}}\left(Y_{\mathrm{s}}|X\right)= \int_{\mathbb{C}}\sum_{l = 0}^{3}\xi\left(p_{Y_{\mathrm{s}}|X}(y_{l}|x)\right)\mathrm{d}F_{X} \label{eqn:H(Ys|X)}
\end{align}
while the second entropy term can be expressed as
\begin{align}
\mathbb{H}_{F_{X}}\left(Y_{\mathrm{s}}|H_{\mathrm{s}},X\right)= \iint\limits_{\mathbb{C}\; \mathbb{C}}\sum_{l = 0}^{3}\xi\left(W_{y_{l}}^{(\mathrm{s})}(h_{\mathrm{s}} x)\right)\mathrm{d}F_{H_{\mathrm{s}}}\mathrm{d}F_{X}.\label{eqn:H(Y_s|G,X)}
\end{align}
Note that $\mathbb{I}\left(\eta;Y_{\mathrm{s}}|X\right) = \mathbb{I}\left(H_{\mathrm{s}};Y_{\mathrm{s}}|X\right)$ is due to the one-to-one mapping between $\eta$ and $H_{\mathrm{s}}^n$ (See \cite[Lemma 1]{Liu:2023}).

Suppose we define the set $\Omega_{P}$ to be the set of all input distributions $F_{X}$ that satisfy the average power constraint $P$. Given the definitions of CMI and SMI, it then follows that the communication and sensing capacities can be written as
\begin{align}
    C_{\mathrm{comm}}(P) = \max_{F_{X}\in \Omega_{P} }\; \mathbb{I}\left(X;Y_{\mathrm{c}}|H_{\mathrm{c}}\right)
\end{align}
and
\begin{align}
    C_{\mathrm{sense}}(P) = \max_{F_{X}\in \Omega_{P} }\; \mathbb{I}\left(H_{\mathrm{s}};Y_{\mathrm{s}}|X\right),
\end{align}
respectively. In the next section, we show that $C_{\mathrm{comm}}(P)$ and $C_{\mathrm{sense}}(P)$ are achieved by the same input distribution for any average power constraint $P$.

\subsection{Characterization of the Communication-Sensing Capacity Region}

We now identify the properties of the input distribution that simultaneously attains the communication capacity $C_{\mathrm{comm}}$ and sensing capacity $C_{\mathrm{sense}}$. To lay the groundwork for the derivation of these properties, we first present a key lemma that establishes the symmetry in the communication and sensing channel laws.
\begin{lemma}\label{lemma:channel_symmetry}
    For any $k\in \mathbb{Z}$, the function $W^{(\mathrm{a})}_{y_{l}}\left(z\right)$ satisfies
    \begin{align}
        W^{(\mathrm{a})}_{y_{l\oplus k}}\left(z\right)=W^{(\mathrm{a})}_{y_l}\left(z e^{-j\frac{\pi k}{2}}\right),
    \end{align}
    where $\oplus$ is a modulo-4 addition.
\end{lemma}
\begin{proof}
    See Appendix \ref{proof:lemma1}.
\end{proof}
Intuitively, applying a $-\frac{\pi k}{2}$ phase rotation for some integer $k$ to the channel input $z$ corresponds to a modulo shift of $k$ to the output of the 1-bit I/Q quantizer. Next, we introduce a class of input distributions to which any optimal distribution $F_X^* \in \Omega_P$ belongs.
\begin{definition}\label{definition:pi_2_sym_dist}
    A distribution $F_{X}$ is said to be $\frac{\pi}{2}$-symmetric if
    \begin{align}
        F_{X}(x) \sim F_{X}(x e^{j\frac{\pi}{2}k})\qquad\forall k\in\mathbb{Z}.
    \end{align}
    Simply put, $F_{X}$ is invariant under any integer multiple of $\frac{\pi}{2}$ phase shift.
\end{definition}

Any distribution $F_{X}$ can be transformed to a $\frac{\pi}{2}$-symmetric distribution using the transformation
\begin{align}
    F_{X}^{\mathrm{s}}(x) = \frac{1}{4}\sum_{k = 0}^{3}F_{X}(x e^{j\frac{\pi}{2}k}). 
\end{align}
This transformation is used to establish the communication-sensing capacity region. Given Lemma \ref{lemma:channel_symmetry} and Definition \ref{definition:pi_2_sym_dist}, we are now ready to state the first main result of this paper.

\begin{theorem}\label{theorem:cap_region}
    Suppose the fading coefficient $H_{\mathrm{a}}\sim \mathcal{N}_{\mathbb{C}}(0,1)$, where $\mathrm{a}\in\{\mathrm{c},\mathrm{s}\}$, can be expressed in polar form as $H_{\mathrm{a}} = \sqrt{\Gamma_{\mathrm{a}}}e^{j\Theta_{\mathrm{a}}}$. Suppose further that $H_{\mathrm{c}}$ is available at the communication receiver. The communication capacity and sensing capacity of the 1-bit Gaussian ISAC channel are 
    \begin{align}\label{eq:C_comm}
        C_{\mathrm{comm}}(P) =& 2 - \mathbb{E}_{\Gamma_{\mathrm{c}},\Theta_{\mathrm{c}}}\Bigg\{\mathbb{H}_{\mathrm{b}}\left(Q\left(\sqrt{\frac{\Gamma_{\mathrm{c}} P\cos^2(\Theta_{\mathrm{c}})}{\sigma_{\mathrm{c}}^2/2}}\right)\right)\nonumber\\
        &\quad\quad+\mathbb{H}_{\mathrm{b}}\left(Q\left(\sqrt{\frac{\Gamma_{\mathrm{c}} P\sin^2(\Theta_{\mathrm{c}})}{\sigma_{\mathrm{c}}^2/2}}\right)\right)\Bigg\}
    \end{align}
    and
     \begin{align}\label{eq:C_sense}
        C_{\mathrm{sense}}(P) =& 2 - \mathbb{E}_{\Gamma_{\mathrm{s}},\Theta_{\mathrm{s}}}\Bigg\{\mathbb{H}_{\mathrm{b}}\left(Q\left(\sqrt{\frac{\Gamma_{\mathrm{s}} P\cos^2(\Theta_{\mathrm{s}})}{\sigma_{\mathrm{s}}^2/2}}\right)\right)\nonumber\\
        &\quad\quad+\mathbb{H}_{\mathrm{b}}\left(Q\left(\sqrt{\frac{\Gamma_{\mathrm{s}} P\sin^2(\Theta_{\mathrm{s}})}{\sigma_{\mathrm{s}}^2/2}}\right)\right)\Bigg\},
    \end{align}
    respectively, where $\mathbb{H}_{\mathrm{b}}\left(p\right)$ with $p\in[0,1]$ is the binary entropy function. The communication and sensing capacities are simultaneously attained by any $\frac{\pi}{2}$-symmetric distribution with amplitude $\sqrt{P}$. Consequently, the communication-sensing capacity region is the set of all communication rate and sensing rate pairs $(R_{\mathrm{c}},R_{\mathrm{s}})$ that satisfy
    \begin{subequations}\label{eq:comm_sensing_region}
    \begin{align}
        0 \leq R_{\mathrm{c}} \leq& C_{\mathrm{comm}}(P)\\
        0 \leq R_{\mathrm{s}} \leq& C_{\mathrm{sense}}(P).
    \end{align}
    \end{subequations}
\end{theorem}
\begin{proof}
    See Appendix \ref{proof:theorem1}.
\end{proof}

A key insight from this result is that, under Gaussian ISAC channel with i.i.d. fading, the communication capacity $C_{\mathrm{comm}}(P)$ and sensing capacity $C_{\mathrm{sense}}(P)$ of the 1-bit Gaussian ISAC channel can be achieved simultaneously. In other words, there is no inherent trade-off between communication and sensing under this scenario. Moreover, the input distribution that achieves both $C_{\mathrm{comm}}(P)$ and $C_{\mathrm{sense}}(P)$ is not unique. Any $2^b$-phase shift keying with $b \geq 2$ and amplitude $\sqrt{P}$ achieve both capacities. More generally, these capacities can also be attained by continuous input distributions, such as a uniform distribution over a circle of radius $\sqrt{P}$.

\section{Optimal Power Control Strategy in 1-bit ISAC System under CSIT}
\label{section:power_alloc}

\subsection{Formulation of the Power Control Policy}\label{subsection:P_opt_formulation}
The previous section established that any $\frac{\pi}{2}$-symmetric distribution  with single amplitude level of $\sqrt{P}$ simultaneously achieves the communication and sensing capacities when the communication CSI $H_{\mathrm{c}}$ is available at the receiver. We now consider the case when the communication CSI is also available at the transmitter. Under this scenario, the ISAC transmitter can adapt its transmission according to the instantaneous state of the communication channel. This enables the transmitter to rotate $X$ by $-\Theta_{\mathrm{c}}$ to compensate for the phase of $H_{\mathrm{c}}$. Consequently, an achievable communication rate can be expressed as
\begin{align}\label{eq:C_comm_CSIT}
C_{\mathrm{comm}}^{(\mathrm{CSIT})}\left(P_{\gamma_{\mathrm{c}}}\right) = 2 - 2 \mathbb{E}_{\Gamma_{\mathrm{c}}}\Bigg\{\mathbb{H}_{\mathrm{b}}\left(Q\left(\sqrt{\frac{\Gamma_{\mathrm{c}} P_{\Gamma_{\mathrm{c}}}}{\sigma_{\mathrm{c}}^2}}\right)\right)\Bigg\}
\end{align}
and is attained by a rotated QPSK modulation scheme with amplitude $\sqrt{P_{\gamma_{\mathrm{c}}}}$ \cite{Krone_2010,bernardo2022_phasequant}. The function $P_{\gamma_{\mathrm{c}}}$ is the power control policy, which depends on the instantaneous communication channel gain $\gamma_{\mathrm{c}} = |h_{\mathrm{c}}|^2$ and satisfies $\int_{0}^{\infty}f_{\Gamma_{\mathrm{c}}}(\gamma_{\mathrm{c}})\cdot P_{\gamma_{\mathrm{c}}}\;\mathrm{d}\gamma_{\mathrm{c}} \leq P$. Note that the expectation over the communication channel phase is absent in \eqref{eq:C_comm_CSIT} due to the pre-rotation by $-\Theta_{\mathrm{c}}$.

\begin{figure}
\centering
    \includegraphics[width=1.05\linewidth]{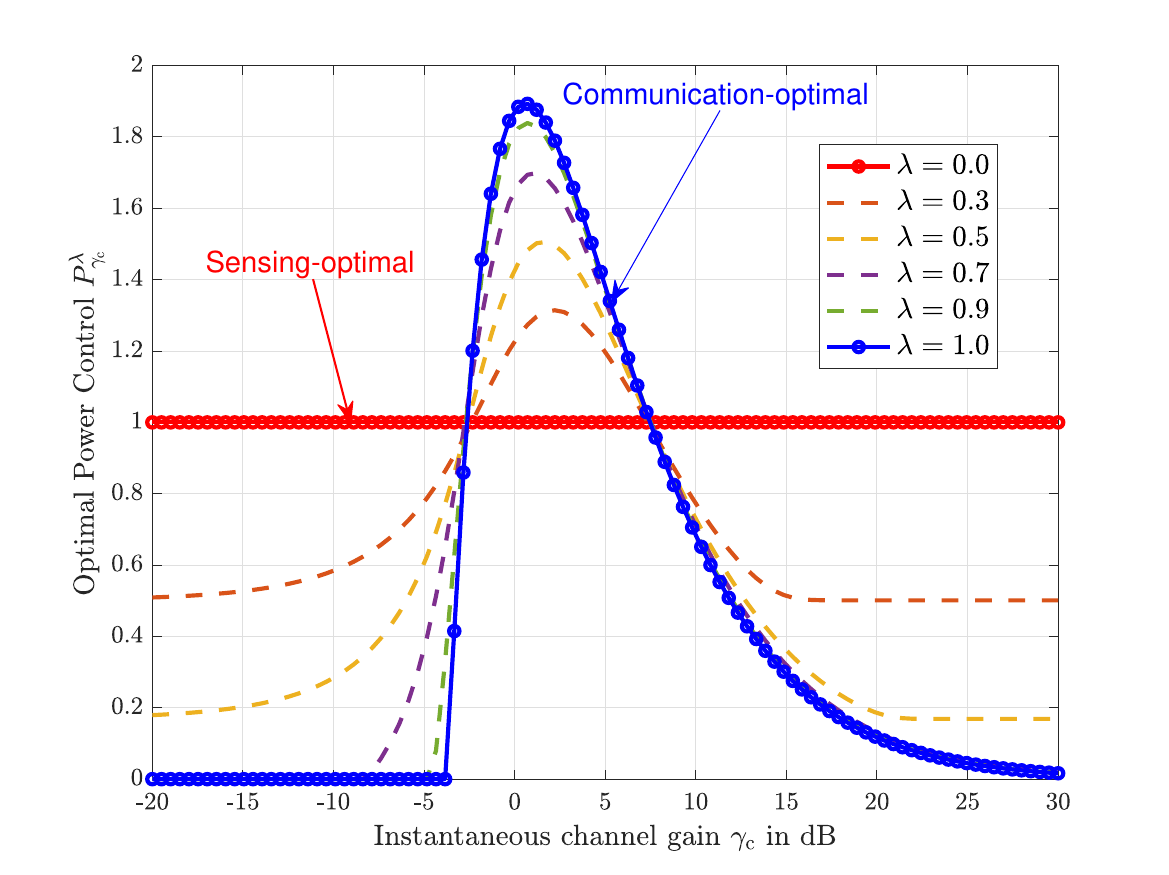}
    \caption{Optimal power control $P^{\lambda}_{\gamma_{\mathrm{c}}}$ vs. instantaneous communication channel gain $\gamma_{\mathrm{c}}$ at 0 dB SNR for both communication and sensing channels.}
    \label{fig:pc_plot}
\end{figure}

The power control policy $P_{\gamma_{\mathrm{c}}}$ can be optimized to attain the communication capacity under CSIT. The optimal power control policy derived in \cite{bernardo2022_phasequant} can be specialized to the 1-bit I/Q quantizer to attain the ergodic capacity of a 1-bit Rayleigh fading channel with CSIT. However, when the objective also accounts for the sensing performance, this communication-optimal power control policy is no longer optimal. To capture the design of the power control scheme for joint communication and sensing, we consider an optimization problem that maximizes the weighted sum of CMI and SMI under an average transmit power constraint. Specifically, we formulate the objective function
\begin{align}
    \mathcal{C}_{\lambda}\left(P_{\gamma_{\mathrm{c}}}\right) = \lambda \cdot C_{\mathrm{comm}}^{\mathrm{(CSIT)}}\left(P_{\gamma_{\mathrm{c}}}\right) + (1-\lambda)\cdot C_{\mathrm{sense}}\left(P_{\gamma_{\mathrm{c}}}\right), \label{eq:obj_func}
\end{align}
where $\lambda\in[0,1]$ controls the relative priority between communication and sensing tasks. For a given $\lambda$, we seek the optimal power control policy, denoted $P_{\gamma_{\mathrm{c}}}^*$, that maximizes the objective function $\mathcal{C}_{\lambda}\left(P_{\gamma_{\mathrm{c}}}\right)$ while satisfying the average transmit power constraint. The following theorem presents the optimal power control policy.
\begin{theorem}\label{theorem:power_alloc}
    Define the function
    \begin{align}
        D\left(k,\beta\right) = -\sqrt{\frac{k}{8\pi \beta}} e^{-\frac{k \beta}{2}}\log\left(\frac{1 - Q(\sqrt{k\beta})}{Q(\sqrt{k\beta}}\right), \label{eq:D_func}
    \end{align}
    which is the derivative of $\mathbb{H}_{\mathrm{b}}\left(Q\left(\sqrt{k\beta}\right)\right)$ with respect to $\beta$. Define another function $G_{\lambda,\gamma_{\mathrm{c}}}\left(\alpha\right)$ whose full expression is written in \eqref{eq:G_lambda} and is invertible with respect to $\alpha$ for fixed $\mu$ and $\gamma_{\mathrm{c}}$. Then, the power control policy that maximizes $\mathcal{C}_{\lambda}\left(P_{\gamma_{\mathrm{c}}}\right)$ for a given $\lambda\in [0,1]$ is
\begin{figure*}[t!]
    \begin{align}\label{eq:G_lambda}
        G_{\lambda,\gamma_{\mathrm{c}}}\left(\alpha\right) = 2\lambda\cdot D\left(\frac{\gamma_{\mathrm{c}}}{\sigma_{\mathrm{c}}^2},\alpha\right) + (1-\lambda)\cdot \mathbb{E}_{\Gamma_{\mathrm{s}},\Theta_{\mathrm{s}}}\left\{D\left(\frac{\Gamma_{\mathrm{s}}\cos^2(\Theta_{\mathrm{s}})}{\sigma_{\mathrm{s}}^2/2},\alpha\right) + D\left(\frac{\Gamma_{\mathrm{s}}\sin^2(\Theta_{\mathrm{s}})}{\sigma_{\mathrm{s}}^2/2},\alpha\right)\right\}
    \end{align}
    \hrulefill
\end{figure*}
    \begin{align}\label{eq:P_opt}
        P_{\gamma_{\mathrm{c}}}^{\lambda} = \left[G_{\lambda,\gamma_{\mathrm{c}}}^{-1}\left(\mu\right)\right]^+,
    \end{align} 
    where $\mu$ satisfies 
    \begin{align}
        \int_{0}^{\infty} [G_{\lambda,\gamma_{\mathrm{c}}}^{-1}(\mu)]^+\cdot f_{\Gamma_{\mathrm{c}}}(\gamma_{\mathrm{c}})\;\mathrm{d}\gamma_{\mathrm{c}} \leq P. \label{eq:p_constraint}
    \end{align}
\end{theorem}
\begin{proof}
    See Appendix \ref{proof:theorem2}
\end{proof}
Theorem \ref{theorem:power_alloc} characterizes the optimal power control policy through necessary and sufficient optimality conditions, revealing how the weighted contributions of communication and sensing jointly determine the power allocation given the communication CSI. In the next section, we provide numerical results that illustrate how the weighting parameter $\lambda$ influences the power control policy in Theorem \ref{theorem:power_alloc}.

\subsection{Numerical Results}\label{subsection:numerical_results}

In this section, we examine the impact of the weighting parameter $\lambda$ on the optimal power control policy derived in Theorem \ref{theorem:power_alloc}. Since the optimal power control policy is given implicitly through the inverse function of $G_{\lambda,\gamma_{\mathrm{c}}}(\alpha)$ and does not admit a closed-form expression, we compute the power allocation numerically. In particular, for a fixed $\lambda$, we first determine the Lagrange multiplier $\mu$ that satisfies the average power constraint in \eqref{eq:p_constraint}, and then solve for $P^{\lambda}_{\gamma_\mathrm{c}}$ numerically using a nested bisection method.

Figure \ref{fig:pc_plot} depicts the optimal power control policy for ISAC with 1-bit I/Q ADC as a function of $\gamma_{\mathrm{c}}$ under different $\lambda$. The noise powers of the communication and sensing channels are set to $\sigma_{\mathrm{c}}^2 = \sigma_{\mathrm{s}}^2 = 1$ and the average power constraint is set to $P = 1$, i.e., 0 dB signal-to-noise ratio (SNR). At $\lambda = 1$ (communication-optimal), there exists a cut-off value of $\gamma_{\mathrm{c}}$ in the power control policy, below which the ISAC transmitter does not transmit anything. Moreover, the power control policy does not allocate power to communication fading states with extremely high $\gamma_{\mathrm{c}}$. This can be attributed to the fact that the capacity is bounded above by 2 bits/cu in the high SNR regime. Thus, there is little or no benefit in allocating power on channel states with large $\gamma_{\mathrm{c}}$. This behavior is reminiscent of the power allocation derived for phase-quantized communication receivers \cite{bernardo2022_phasequant}. As $\lambda$ decreases, the power allocation strategy begins to assign power to both very low and very high values of $\gamma_{\mathrm{c}}$, gradually approaching the sensing-optimal scheme, i.e., no power adaptation given $\gamma_{\mathrm{c}}$. The computational complexity of the power adaptation depends on $\gamma_{\mathrm{c}}$ where a practical approach is to precompute $\mu$ and perform the bisection for every SNR to form a power allocation lookup table depending on $\lambda$. 

Figures \ref{fig:Rc_CSIT} and \ref{fig:Rs_CSIT} illustrate the communication rates and sensing rates, respectively, for $\lambda = 0$ (sensing-optimal) $\lambda = 0.5$ (equal priority), and $\lambda = 1.0$ (communication-optimal). The SNR is varied by fixing $\sigma_{\mathrm{c}}^2 = 1$ and $\sigma_{\mathrm{s}}^2 = 1$ and then sweeping the average power constraint $P$ accordingly. It can be observed that the communication rate under CSIT and $\lambda = 0$ exceeds the communication rate under CSIR. The gap between the two is more pronounced in the high SNR regime. This is due to the significant gains obtained by pre-rotating the transmission to combat $\Theta_{\mathrm{c}}$. Additional increase in communication rate is achieved as $\lambda \rightarrow 1$, reflecting the increased prioritization of the power control policy in the communication performance. Meanwhile, the sensing rate $R_{\mathrm{s}}$ under CSIT and $\lambda = 0$ coincides with the CSIR sensing capacity. This is expected since the sensing objective does not exploit variations in the communication channel gain. When $\lambda = 1$, the sensing rate falls below the CSIR sensing capacity and saturates at $\approx1.7$ bits/cu in the high SNR regime. This can be attributed to the fact that $R_{\mathrm{c}}$ has already achieved the 2 bits/cu upper bound around 20 dB SNR when $\lambda = 1$. With zero weight assigned to the sensing objective, the objective function $\mathcal{C}_{\lambda = 1}(P_{\gamma_{\mathrm{c}}})$ cannot be further improved by increasing the sensing rate. This saturation can be avoided by selecting $\lambda < 1$, as shown in the $\lambda = 0.999$ curve in Figure \ref{fig:Rs_CSIT}. Finally, we observe that assigning equal priority to sensing and communication objectives achieves near-optimal performance for both $R_{\mathrm{c}}$ and $R_{\mathrm{s}}$.

\section{Conclusion}
\label{section:conclusion}

\begin{figure}[t]
\centering
    \includegraphics[width=1.05\linewidth]{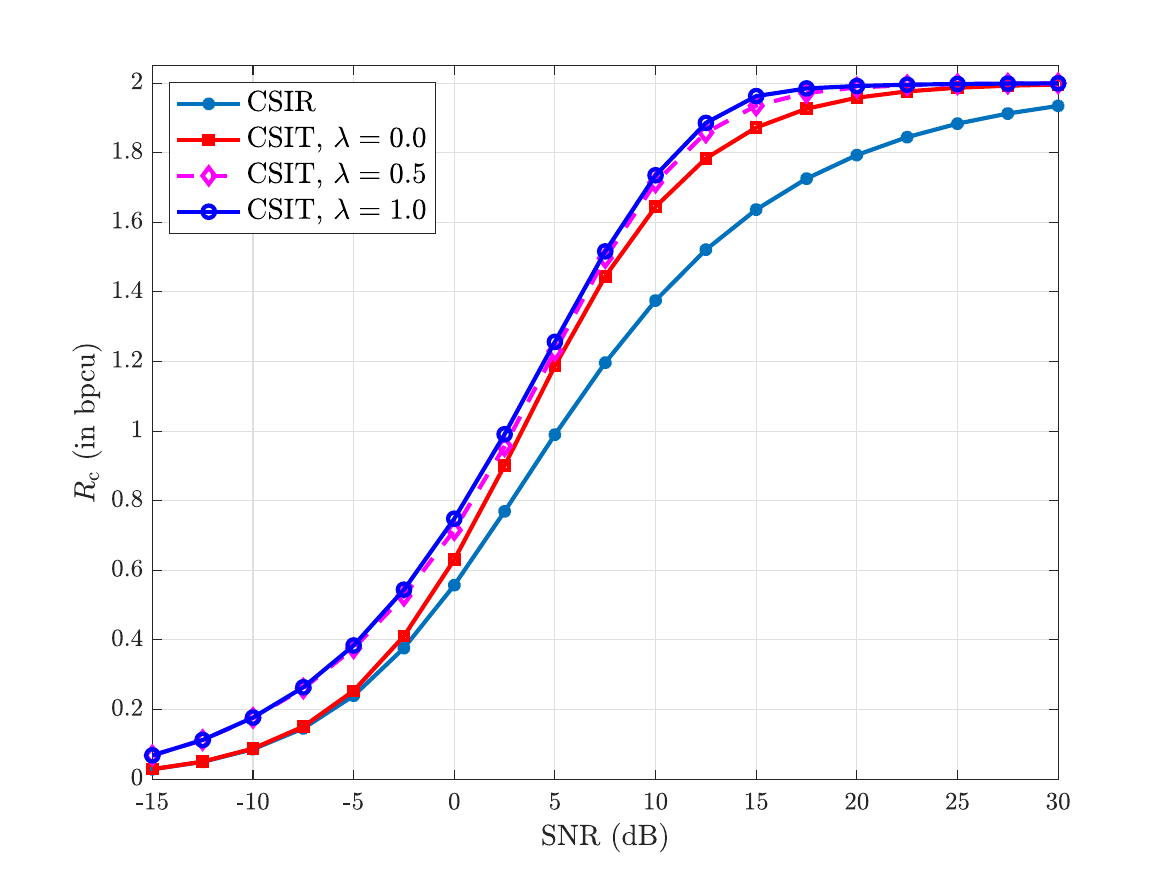}
    \caption{Communication rate vs. SNR for different $\lambda$ under CSIT. Superimposed in this plot is the communication capacity under CSIR.}
    \label{fig:Rc_CSIT}
\end{figure}
\begin{figure}[t]
\centering
    \includegraphics[width=1.05\linewidth]{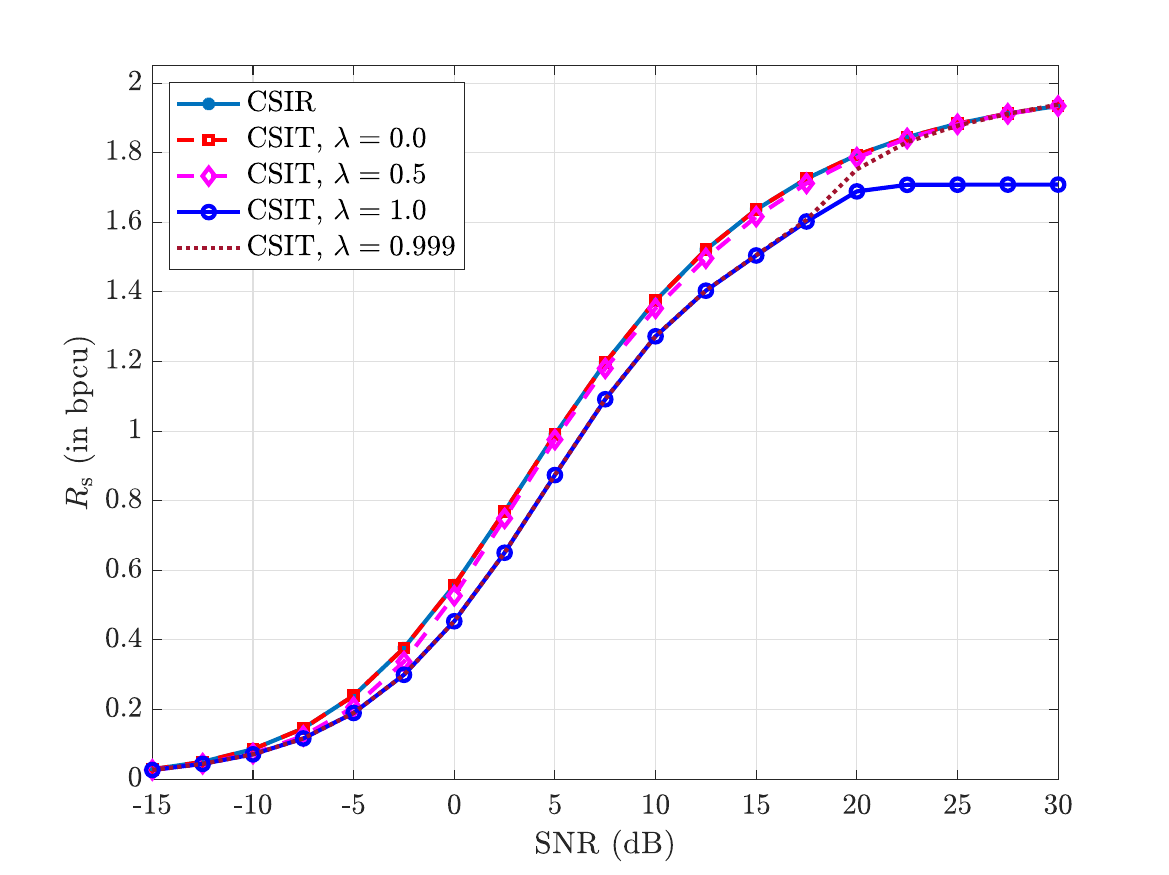}
    \caption{Sensing rate vs. SNR for different $\lambda$ under CSIT. Superimposed in this plot is the sensing capacity under CSIR.}
    \label{fig:Rs_CSIT}
\end{figure}

In this paper, we investigated the fundamental limits of 1-bit ISAC systems. Under the communication CSIR scenario, the communication–sensing capacity region is characterized and we showed that there is no inherent trade-off between communication and sensing performance. In particular, both CMI and SMI can be simultaneously maximized by a $\frac{\pi}{2}$-symmetric input distribution with amplitude $\sqrt{P}$. Under the communication CSIT scenario, we showed that the ability to adapt the transmission introduces a non-trivial trade-off between communication and sensing. By formulating a weighted MI optimization problem, we derived the optimal power control policy and demonstrated that it deviates significantly from the classical water-filling strategy due to the sensing requirement. Numerical results illustrated how the optimal policy evolves with the weighting parameter, transitioning from a communication-centric allocation to a more uniform sensing-oriented strategy. These findings highlight the critical role of communication CSI in shaping ISAC system design under coarse quantization. Future work may extend this framework to investigate the tradeoff from an estimation-theoretic sensing metric, multi-antenna systems, multi-user settings, and more general quantization architectures.


\begin{appendices}
 \section{Proof of Lemma \ref{lemma:channel_symmetry}}
 \label{proof:lemma1}
 
Note that the dependence of $W_{y_l}^{(\mathrm{a})}(z)$ on $z$ and $y_l$ is through the terms $\Re\{z\}\Re\{y_l\}$ and $\Im\{z\}\Im\{y_l\}$. Let $z = \rho_{z}e^{j\theta_{z}}$ be the polar form of $z$ and let $y_l = \cos\left(\frac{\pi}{4} + \frac{\pi l}{2}\right) + j\sin\left(\frac{\pi}{4} + \frac{\pi l}{2}\right)$ be the Cartesian form of $y_{l}$. We can reduce the term $\Re\{z e^{-j\frac{\pi k}{2}}\}\Re\left\{y_l\right\}$ into
\begingroup
\allowdisplaybreaks
\begin{align}\label{eq:real_part}
    &\Re\{z e^{-j\frac{\pi k}{2}}\}\Re\left\{y_l\right\}= \rho_{z}\cos\left(\theta_z - \frac{\pi k}{2}\right)\cos\left(\frac{\pi}{4} + \frac{\pi l}{2}\right)\nonumber\\
     &\quad= \frac{\rho_{z}\left[\cos\left(\theta_z\right)\cos\left(\frac{\pi k}{2}\right) + \sin\left(\theta_z\right)\sin\left(\frac{\pi k}{2}\right)\right]\cos\left(\frac{\pi}{4} + \frac{\pi l}{2}\right)}{2}\nonumber\\
     &\quad=\begin{cases}
         \rho_{z}\cos(\theta_{z})\cos\left(\frac{\pi}{4}+\frac{\pi(l+k)}{2}\right),\qquad \text{$k$ is even}\\
         \rho_{z}\sin(\theta_{z})\sin\left(\frac{\pi}{4}+\frac{\pi(l+k)}{2}\right),\qquad \text{$k$ is odd}
     \end{cases}\nonumber\\
     &\quad=\begin{cases}
         \Re\{z\}\Re\{y_{l\oplus k}\},\qquad \text{$k$ is even}\\
         \Im\{z\}\Im\{y_{l\oplus k}\},\qquad \text{$k$ is odd}
     \end{cases}.
\end{align}
Similarly, the term $\Im\{z\}\Im\{y_l\}$ can be simplified to
\begin{align}\label{eq:imag_part}
    &\Im\{z e^{-j\frac{\pi k}{2}}\}\Im\left\{y_l\right\}= \rho_{z}\sin\left(\theta_z - \frac{\pi k}{2}\right)\sin\left(\frac{\pi}{4} + \frac{\pi l}{2}\right)\nonumber\\
     &\quad= \frac{\rho_{z}\left[\sin\left(\theta_z\right)\cos\left(\frac{\pi k}{2}\right) - \cos\left(\theta_z\right)\sin\left(\frac{\pi k}{2}\right)\right]\sin\left(\frac{\pi}{4} + \frac{\pi l}{2}\right)}{2}\nonumber\\
     &\quad=\begin{cases}
         \rho_{z}\sin(\theta_{z})\sin\left(\frac{\pi}{4}+\frac{\pi(l+k)}{2}\right),\qquad \text{$k$ is even}\\
         \rho_{z}\cos(\theta_{z})\cos\left(\frac{\pi}{4}+\frac{\pi(l+k)}{2}\right),\qquad \text{$k$ is odd}
     \end{cases}\nonumber\\
     &\quad=\begin{cases}
         \Im\{z\}\Im\{y_{l\oplus k}\},\qquad \text{$k$ is even}\\
         \Re\{z\}\Re\{y_{l\oplus k}\},\qquad \text{$k$ is odd}
     \end{cases}.
\end{align}
\endgroup
The proof is completed by applying \eqref{eq:real_part} and \eqref{eq:imag_part} into \eqref{eq:gen_channel_law_1bit}.

\section{Proof of Theorem \ref{theorem:cap_region}}

\begingroup
\allowdisplaybreaks
Let $Z = H_{\mathrm{a}}\cdot X$ with distribution $F_{Z}(z)$ and define the distribution $F_{Z}^{\mathrm{s}}$ to be
\begin{align}
    F_{Z}^{\mathrm{s}}(z) = \frac{1}{4}\sum_{i = 0}^{3}F_{Z}(ze^{j\frac{\pi k}{2}}),
\end{align}
which is a $\frac{\pi}{2}$-symmetric distribution derived from $F_{Z}$. The output PMF of $Y_{\mathrm{a}}$ induced by $Z\sim F_{Z}^{\mathrm{s}}(z)$ on $W_{y_l}^{(\mathrm{a})}(z)$ is
\begin{align}
    p_{Y_{a}}(y_l) =& \int_{\mathbb{C}}W_{y_l}^{(\mathrm{a})}(z)\;\mathrm{d}F_{Z}^{\mathrm{s}}= \frac{1}{4}\int_{\mathbb{C}}\sum_{k = 0}^{3}W_{y_l}^{(\mathrm{a})}\left(z e^{-j\frac{\pi k}{2}}\right)\mathrm{d}F_{Z}\nonumber\\
    =& \frac{1}{4}\int_{\mathbb{C}}\underbrace{\sum_{k = 0}^{3}W_{y_{l\oplus k}}^{(\mathrm{a})}(z)}_{=1}\mathrm{d}F_{Z}(z) = \frac{1}{4}
\end{align}
for all $y_l\in\mathcal{Y}$. The second equality is obtained using the relationship between $F_{Z}^{\mathrm{s}}$ and $F_{Z}$. The third equality is obtained using Lemma \ref{lemma:channel_symmetry}. As a result, $p_{Y_{a}}(y_l)$ is uniformly distributed when $Z$ is a $\frac{\pi}{2}$-symmetric distribution and the corresponding output entropy is maximized, i.e., $\mathbb{H}_{F_{Z}^{\mathrm{s}}}\left(Y_{\mathrm{a}}\right) = 2$. Simultaneously, the conditional entropy induced by $F_{Z}^{\mathrm{s}}$, denoted $\mathbb{H}_{F_{Z}^{\mathrm{s}}}(Y_{\mathrm{a}}|Z)$, can be reduced to
\begin{align}\label{eq:conditional_ent_sym}
   \mathbb{H}_{F_{Z}^{\mathrm{s}}}(Y_{\mathrm{a}}|Z)&= \sum_{l = 0}^{3}\int_{\mathbb{C}}\xi\left(W_{y_l}^{(\mathrm{a})}\left(z\right)\right)\;\mathrm{d}F_{Z}^{\mathrm{s}}\nonumber\\
    &= \frac{1}{4}\sum_{l = 0}^{3}\sum_{k = 0}^{3}\int_{\mathbb{C}}\xi\left(W_{y_l}^{(\mathrm{a})}\left(z e^{-j\frac{\pi k}{2}}\right)\right)\;\mathrm{d}F_{Z}\nonumber\\
    &= \frac{1}{4}\sum_{l = 0}^{3}\sum_{k = 0}^{3}\int_{\mathbb{C}}\xi\left(W_{y_{l\oplus k}}^{(\mathrm{a})}\left(z \right)\right)\;\mathrm{d}F_{Z}\nonumber\\
    &= \sum_{l' = 0}^{3}\int_{\mathbb{C}}\xi\left(W_{y_{l'}}^{(\mathrm{a})}\left(z \right)\right)\mathrm{d}F_{Z}\nonumber\\
    &= \mathbb{H}_{F_{Z}}(Y_{\mathrm{a}}|Z),
\end{align}
where we introduce the variable $l' = l\oplus k$ in the fourth equality. Since entropy is invariant of the labeling, we get \eqref{eq:conditional_ent_sym}.
Consequently, this implies that $\mathbb{I}\left(Z_1;Y_{\mathrm{a}}\right) \geq \mathbb{I}\left(Z_2;Y_{\mathrm{a}}\right)$ for $Z_{1}\sim F_{Z}^{\mathrm{s}}$ and $Z_2\sim F_{Z}$ and the capacity-achieving $Z$ should be $\frac{\pi}{2}$-symmetric. The $\frac{\pi}{2}$-symmetry on $Z = H_{\mathrm{c}}\cdot X$ is attained in the communication channel for any fixed $H_{\mathrm{c}}$ when $F_{X}$ is a $\frac{\pi}{2}$-symmetric distribution. Meanwhile, the $\frac{\pi}{2}$-symmetry on $Z = H_{\mathrm{s}}\cdot X$ is attained in the sensing channel for any fixed $X$ since $H_{\mathrm{s}}\sim \mathcal{N}_{\mathbb{C}}(0,1)$ is also $\frac{\pi}{2}$-symmetric distribution.

When $F_{X}$ is a $\frac{\pi}{2}$-symmetric distribution, the capacity boils down to finding the $\frac{\pi}{2}$-symmetric input distribution $F_{X}$ in $\Omega_{P}$ such that $\mathbb{H}(Y_{\mathrm{a}}| Z)$ (or equivalently, $\mathbb{H}(Y_{\mathrm{a}}| H_{\mathrm{a}}, X)$) is minimized. Since the real and imaginary part of $Y_{\mathrm{a}}$ are independent conditioned on $Z = \sqrt{\Gamma_{Z}}e^{j\Theta_{Z}}$ under the channel law in \eqref{eq:gen_channel_law_1bit}, the conditional entropy can be written as
\begin{align}
    \mathbb{H}(Y_{\mathrm{a}}|Z) =& \mathbb{E}_{\Theta_{Z},\Gamma_{Z}}\Bigg\{\mathbb{H}_{\mathrm{b}}\left(Q\left(\sqrt{\frac{\Gamma_{Z}\cos^2\Theta_{Z}}{\sigma_{\mathrm{a}}^2/2}}\right)\right)\nonumber\\
    &\qquad\quad+ \mathbb{H}_{\mathrm{b}}\left(Q\left(\sqrt{\frac{\Gamma_{Z}\sin^2\Theta_{Z}}{\sigma_{\mathrm{a}}^2/2}}\right)\right)\Bigg\},
\end{align}
where the binary entropy function is used since the real and imaginary part of $Y_{\mathrm{a}}$ each has two possible outputs and their probabilities are described by the $Q$-function. Using $Z = H_{\mathrm{a}}\cdot X$ with $X = \sqrt{\Gamma_{X}}e^{j\Theta_{X}}$ and $\Theta_{\mathrm{U}} = \Theta_{\mathrm{a}} + \Theta_{X}$, we get
\begin{align}
    &\mathbb{H}(Y_{\mathrm{a}}|H_{\mathrm{a}},X)\nonumber\\
    &=\mathbb{E}_{\Gamma_{\mathrm{a}},\Theta_{\mathrm{U}}}\Bigg\{\mathbb{E}_{\Gamma_{X}|\Theta_{\mathrm{X}}}\Bigg\{\mathbb{H}_{\mathrm{b}}\left(Q\left(\sqrt{\frac{\Gamma_{\mathrm{U}}\Gamma_{X}\cos^2\left(\Theta_{\mathrm{U}}\right)}{\sigma_{\mathrm{a}}^2/2}}\right)\right)\nonumber\\
&\qquad+\mathbb{H}_{\mathrm{b}}\left(Q\left(\sqrt{\frac{\Gamma_{\mathrm{a}}\Gamma_{X}\sin^2\left(\Theta_{\mathrm{U}}\right)}{\sigma_{\mathrm{a}}^2/2}}\right)\right)\Bigg\}\Bigg\}\nonumber\\
&\geq\mathbb{E}_{\Gamma_{\mathrm{a}},\Theta_{\mathrm{a}}}\Bigg\{\mathbb{H}_{\mathrm{b}}\left(Q\left(\sqrt{\frac{\Gamma_{\mathrm{a}}\mathbb{E}_{\Gamma_{X}|\Theta_{\mathrm{X}}}\left\{\Gamma_{X}\right\}\cos^2\left(\Theta_{\mathrm{a}}\right)}{\sigma_{\mathrm{a}}^2/2}}\right)\right)\nonumber\\
&\qquad+\mathbb{H}_{\mathrm{b}}\left(Q\left(\sqrt{\frac{\Gamma_{\mathrm{a}}\mathbb{E}_{\Gamma_{X}|\Theta_{\mathrm{X}}}\left\{\Gamma_{X}\right\}\sin^2\left(\Theta_{\mathrm{a}}\right)}{\sigma_{\mathrm{a}}^2/2}}\right)\right)\Bigg\}.
\end{align}
The inequality is obtained by noting that $\mathbb{H}_{\mathrm{b}}(Q(\sqrt{w}))$ is convex in $w$ \cite{singh_2009} so Jensen's inequality can be applied. Note that equality can be achieved here when $\mathbb{E}_{\Gamma_{X}|\Theta_{\mathrm{X}}}\left\{\Gamma_{X}\right\}$ is constant, i.e. $F_{X}$ has a constant amplitude.  Moreover, $\Theta_{\mathrm{U}}$ and $\Theta_{\mathrm{a}}$ are both uniformly distributed so we can replace $\Theta_{\mathrm{U}}$ by $\Theta_{\mathrm{a}}$ inside the expectation. Since $\mathbb{H}_{\mathrm{b}}(Q(\sqrt{w}))$ is also a decreasing function of $w$, the smallest $\mathbb{H}(Y_{\mathrm{a}}|H_{\mathrm{a}},X)$ is attained by setting $\mathbb{E}_{\Gamma_{X}|\Theta_{\mathrm{X}}}\left\{\Gamma_{X}\right\} = P$, i.e., transmit at full power. 

Specializing the above result to the communication capacity and sensing capacity gives us \eqref{eq:C_comm} and \eqref{eq:C_sense}. Since the same input distribution is optimal for both 1-bit communication and sensing channels, we get the capacity region in \eqref{eq:comm_sensing_region}.


\endgroup

 \label{proof:theorem1}
\section{Proof of Theorem \ref{theorem:power_alloc}}

To maximize the objective function in (\ref{eq:obj_func}) under the average transmit power constraint, we form the Lagrangian
\begin{align}
 \mathcal{L}(P_{\gamma_{\mathrm{c}}},\mu) = \mathcal{C}_\lambda(P_{\gamma_\mathrm{c}}) - \mu \Big(\int_0^\infty P_{\gamma_\mathrm{c}} f_{\Gamma_c}(\gamma_\mathrm{c}) \mathrm{d}\gamma_\mathrm{c}-P \Big),
\end{align}
where $\mu$ is the Lagrangian multiplier. The functions $\mathrm{C}_{\mathrm{comm}}^{\mathrm{(CSIT)}}(P_{\gamma_{\mathrm{c}}})$ and $\mathrm{C}_{\mathrm{sense}}^{\mathrm{(CSIT)}}(P_{\gamma_{\mathrm{c}}})$ are concave in $P_{\gamma_{\mathrm{c}}}$. It then follows that (\ref{eq:obj_func}) is also concave in $P_{\gamma_{\mathrm{c}}}$. As such, the Karush-Kuhn-Tucker (KKT) conditions are necessary and sufficient to establish the optimality of a particular $P_{\gamma_{\mathrm{c}}}$. Taking the functional derivative of $\mathcal{L}(P_{\gamma_{\mathrm{c}}},\mu)$ with respect to the power control policy $P_{\gamma_\mathrm{c}}$ and imposing the stationary condition gives
\begin{align}
\frac{\partial}{\partial P_{\gamma_\mathrm{c}}} \Big[\mathcal{C}_\lambda(P_{\gamma_\mathrm{c}}) - \mu \Big(\int_0^\infty P_{\gamma_\mathrm{c}} f_{\Gamma_c}(\gamma_\mathrm{c}) \mathrm{d}\gamma_\mathrm{c}-P \Big)\Big] = 0.
\end{align}
Using the linearity of expectation, the derivatives of the functionals $C_{\mathrm{comm}}^{(\mathrm{CSIT})}(P_{\gamma_{\mathrm{c}}})$ and $C_{\mathrm{sense}}(P_{\gamma_{\mathrm{c}}})$ with respect to $P_{\gamma_\mathrm{c}}$ can be described in terms of the function $D(k,\beta)$ in \eqref{eq:D_func}. Consequently, the stationary condition reduces to the point-wise condition
\begin{align}
    [G_{\lambda,\gamma_\mathrm{c}}(P_{\gamma_\mathrm{c}})-\mu] f_{\Gamma_\mathrm{c}}(\gamma_c) = 0,
\end{align}
where $G_{\lambda,\gamma_\mathrm{c}}(P_{\gamma_{\mathrm{c}}})$ is defined in \eqref{eq:G_lambda}. Solving for $P_{\gamma_{\mathrm{c}}}$ with the constraint $P_{\gamma_{\mathrm{c}}} > 0$ yields \eqref{eq:P_opt}.


 \label{proof:theorem2}
 \end{appendices}

\section*{Acknowledgements}\label{section:acknowledgement}

E. Trinidad  acknowledges the DOST-Science Education Institute (DOST-SEI) and the Engineering Research and Development for Technology (ERDT) Program for supporting his doctoral studies. N.I. Bernardo acknowledges the Office of the Chancellor of the University of the
Philippines Diliman, through the Office of the Vice Chancellor for Research and
Development, for funding support through the PhD Incentive Grant (Project No.
262629 PhDIA Y2).

\renewcommand*{\bibfont}{\footnotesize}
\begingroup
\footnotesize  
\printbibliography
\endgroup






\end{document}